\documentclass[a4paper,UKenglish,cleveref,autoref,numberwithinsect,thm-restate]{lipics-v2021}

\hideLIPIcs  

\usepackage{mathtools}
\usepackage{comment}
\usepackage{xcolor}

\title{Linear-Time Exact Computation of Influence Spread on Bounded-Pathwidth Graphs} 

\titlerunning{Linear-Time Exact Influence Spread on Bounded-Pathwidth Graphs} 

\author{Kengo Nakamura}{Communication Science Laboratories, NTT, Inc., Japan \and \url{https://www.kecl.ntt.co.jp/icl/lirg/members/nakamura/} }{kengo.nakamura@ntt.com}{https://orcid.org/0000-0002-9615-3479}{}

\author{Masaaki Nishino}{Communication Science Laboratories, NTT, Inc., Japan \and \url{https://sites.google.com/view/masaaki-nishino} }{masaaki.nishino@ntt.com}{https://orcid.org/0000-0001-6489-5446}{}

\authorrunning{K. Nakamura and M. Nishino} 

\Copyright{Kengo Nakamura and Masaaki Nishino} 



\category{} 

\funding{This work was supported by JSPS KAKENHI Grant Number JP26K02906.}

\nolinenumbers 

\newcommand{\reachto}{\!\rightsquigarrow\!}

\newcommand{\prob}[1]{\text{\rm Pr}( #1 )}
\newcommand{\probi}[1]{\text{\rm Pr}_{<i}( #1 )}
\newcommand{\order}[1]{O( #1 )}

\newcommand{\bddlo}{\mathsf{lo}}
\newcommand{\bddhi}{\mathsf{hi}}
\newcommand{\bddf}{\mathsf{f}}
\newcommand{\trans}{\Phi}

\newcommand{\modifymap}{\mathsf{change}}
\newcommand{\pluseq}{\mathrel{+}=}

\theoremstyle{definition}
\newtheorem{problem}[theorem]{Problem}

\usepackage[linesnumbered,ruled,vlined]{algorithm2e}
\DontPrintSemicolon
\SetInd{0.3em}{0.6em}
\SetVlineSkip{0.0pt}
\SetKwData{Pmap}{p}\SetKwData{Qmap}{q}\SetKwData{Rmap}{r}\SetKwData{Resmap}{res}
\colorlet{commentgray}{black!60!}

\SetCommentSty{mycommfont}

\allowdisplaybreaks[4]

\EventEditors{Pierre Fraigniaud}
\EventNoEds{1}
\EventLongTitle{20th Scandinavian Symposium on Algorithm Theory (SWAT 2026)}
\EventShortTitle{SWAT 2026}
\EventAcronym{SWAT}
\EventYear{2026}
\EventDate{June 17--19, 2026}
\EventLocation{Copenhagen, Denmark}
\EventLogo{}
\SeriesVolume{370}
\ArticleNo{21}

\begin{document}

\maketitle

\begin{abstract}
Given a network and a set of vertices called seeds to initially inject information, influence spread is the expected number of vertices that eventually receive the information under a certain stochastic model of information propagation.
Under the commonly used independent cascade model, influence spread is equivalent to the expected number of vertices reachable from the seeds on a directed uncertain graph, and the exact evaluation of influence spread offers many applications, e.g., influence maximization.
Although its evaluation is a \#P-hard task, there is an algorithm that can precisely compute the influence spread in $\order{mn\omega_p^2\cdot 2^{\omega_p^2}}$ time, where $\omega_p$ is the pathwidth of the graph.
We improve this by developing an algorithm that computes the influence spread in $\order{(m+n)\omega_p^2\cdot 2^{\omega_p^2}}$ time.
This is achieved by identifying the similarities in the repetitive computations in the existing algorithm and sharing them to reduce computation.
Although similar refinements have been considered for the probability computation on undirected uncertain graphs, a greater number of similarities must be leveraged for directed graphs to achieve linear time complexity.
\end{abstract}

\section{Introduction}
\label{sec:intro}
Given a network and \emph{seed set} $S$, which is a set of vertices that are initially injected information, \emph{influence spread}~\cite{domingos01im} $\sigma(S)$ is the expected number of vertices that eventually receives the information under a stochastic model of information propagation.
Influence spread is implemented in viral marketing applications with the aim of evaluating the influence of a particular node within a social network.
Currently, it is also a keystone for other network applications such as network monitoring~\cite{leskovec07monitor}, rumor control~\cite{wu17rumor}, target advertisement~\cite{li15target}, and social recommendation~\cite{ye12social}.
Among several information diffusion models, the most basic and commonly used one is the \emph{independent cascade (IC) model}~\cite{goldenburg01ic}.
In this model, a network is modeled as a directed graph $G=(V,E)$, where each edge $e\in E$ is associated with a probability $p_e$.
When a vertex $v$ receives information, it stochastically propagates the information along every outgoing edge $e$ independently with a given probability $p_e$.

The influence spread $\sigma(S)$ under the IC model is equivalent to the expected number of vertices reachable from at least one vertex in $S$ on a directed \emph{uncertain graph}.
Given directed graph $G=(V,E)$ and probability $p_e$ for every edge $e\in E$, we consider an uncertainty in which each edge $e\in E$ is present with probability $p_e$ and absent with probability $1-p_e$.
We assume that each edge's presence or absence is stochastically independent of the other edges.
Accordingly, the probability that a subgraph $(V,E^\prime)\ (E^\prime\subseteq E)$ appears is
\begin{equation}
  \prob{E^\prime}\coloneqq\prod_{e\in E^\prime}p_e\cdot\prod_{e\in E\setminus E^\prime}(1-p_e).
\end{equation}
Here, a directed path from some $s\in S$ to vertex $v$ on the subgraph $(V,E^\prime)$ corresponds to the route of information propagation from $s$ to $v$.
Therefore, given $S\subseteq V$ and $v\in V$, the probability that $v$ receives information is equal to the probability $\prob{S\reachto v}$ that $v$ can be reached from some vertices in $S$ under uncertainty.
Here, $\prob{S\reachto v}=\sum_{E^\prime\in\mathcal{E}_{S\reachto v}}\prob{E^\prime}$, where $\mathcal{E}_{S\reachto v}$ is the family of subsets $E^\prime$ of edges such that $v$ can be reached from some vertices in $S$ on $(G,E^\prime)$.
Finally, the influence spread can be represented as $\sigma(S)=\sum_{v\in V\setminus S}\prob{S\reachto v}$.

\begin{example}
  Figure~\ref{fig:uncertain}a depicts examples of the probabilities $\prob{E^\prime}$ of subgraphs $(V,E^\prime)\ (E^\prime\subseteq E)$.
  For this uncertain graph, there are 15 subgraphs in which $v=4$ is reachable from $\{S\}=1$ (Figure~\ref{fig:uncertain}b); thus, $\prob{\{1\}\reachto 4}$ is the sum of $\prob{E^\prime}$ over all these subgraphs.
  By evaluating $\prob{\{1\}\reachto 2}$, $\prob{\{1\}\reachto 3}$, and $\prob{\{1\}\reachto 4}$ as in Figure~\ref{fig:uncertain}c, $\sigma(\{1\})$ can be computed as their sum.
\end{example}

\begin{figure}[!tb]
    \centering
    \includegraphics[keepaspectratio]{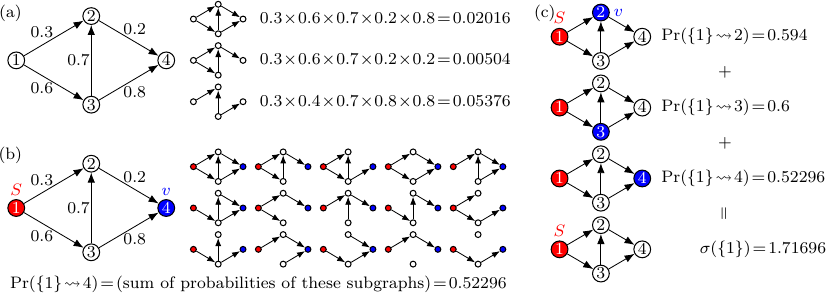}
    \caption{(a) Example of directed uncertain graph and probabilities of subgraphs. (b) All 15 subgraphs in which $v=4$ is reachable from $\{S\}=1$. $\prob{\{1\}\reachto 4}$ is the sum of probabilities of these subgraphs. (c) $\sigma(\{1\})$ is computed as the sum of $\prob{\{1\}\reachto v}$s for $v=2,3,4$.}
    \label{fig:uncertain}
\end{figure}

Since evaluating the influence spread $\sigma(S)$ is a \#P-hard task~\cite{chen10hardness}, most studies have relied on Monte Carlo simulation for evaluating $\sigma(S)$~\cite{goyal11celf,kempe03im,leskovec07monitor,ohsaka14montecarlo}.
However, exact computation of $\sigma(S)$ is important for the following reasons.
First, this allows us to more accurately evaluate the influence spread in larger networks, since real networks often consist of small communities; the exact computation of influence spread in each small community will improve overall evaluation quality.
Second, this allows us to rank the influential vertices in descending order of influence spread.
For this purpose, Monte Carlo simulations are not sufficient because $\Omega(1/\varepsilon^2)$ samples are needed to obtain $(1\pm\varepsilon)$-approximation with high probability~\cite{ohsaka14montecarlo}; to distinguish every vertex's influence spread, high accuracy such as $\varepsilon\sim 10^{-5}$ is sometimes needed, which is costly for Monte Carlo methods.

Maehara et al.~\cite{maehara17} proposed the only non-trivial algorithm to exactly compute $\sigma(S)$ under the IC model.
Given seed set $S$ and $v$, it efficiently computes $\prob{S\reachto v}$, i.e., the probability that $v$ receives information.
Given the path decomposition of $G$ whose width is $\omega_p$, this probability for a single $v$ can be computed in $\order{m\omega_p^2\cdot 2^{\omega_p^2}}$ time, where $m$ is the number of edges.
Hereafter, we mean the pathwidth of a directed graph $G$ in the sense of the pathwidth of the underlying undirected graph of $G$ as used in the literature~\cite{maehara17,suzuki18}.
The influence spread $\sigma(S)=\sum_{v\in V\setminus S}\prob{S\reachto v}$ can be computed in $\order{mn\omega_p^2\cdot 2^{\omega_p^2}}$ time, where $n$ is the number of vertices.

Here, even for bounded-pathwidth graphs, obtaining the $\sigma(S)$ value requires super-linear ($\order{mn}$) time.
Since there are similarities in the repetitive computation for different vertices $v$ using this algorithm, we investigated the use of these similarities to refine their algorithm and thus improve the running time for computing $\sigma(S)$.

\subsection{Our Contribution}
\label{ssec:contribution}
In this paper, we propose an algorithm for computing $\sigma(S)$ for given $S$ under the IC model.
Specifically, the proposed algorithm simultaneously computes $\prob{S\reachto v}$ in an uncertain directed graph for all vertices $v$.
Our main result can be expressed as follows.
\begin{theorem}
  \label{thm:main2}
  Given directed graph $G$, seed set $S$, each edge's probability $p_e$ of presence, and the path decomposition of $G$ whose width is $\omega_p$, the probability $\prob{S\reachto v}$ for every vertex $v$ can be computed in $\order{(m+n)\omega_p^2\cdot 2^{\omega_p^2}}$ time in total.
\end{theorem}

Since even inputting the entire graph requires linear time, this algorithm is asymptotically optimal for graphs with a bounded pathwidth when the constant factor is ignored.
Moreover, the exponential factor of the pathwidth, hidden in the constant, remains the same as in the algorithm of Maehara et al.~\cite{maehara17}: the dependence on the graph sizes is improved from $\order{mn}$ to $\order{m+n}$ while the dependence on the pathwidth is kept the same.

Our algorithm shares certain ideas with the existing algorithm for simultaneously computing the probabilities of connection in undirected graphs~\cite{nakamura21} in that both use similarities in the computation of the previous algorithms~\cite{hardy2007knr,maehara17} for different vertices.
However, there is a non-trivial gap in extending that algorithm~\cite{nakamura21} to directed graphs due to the difference between connectivity in undirected graphs and reachability in directed graphs.
We fill this gap by finding additional similarities in computing $\sigma(S)$ with the previous algorithm~\cite{maehara17}; details are discussed in Section~\ref{sec:discussion}.

\subsection{Related Work}
\label{ssec:related}
\begin{description}
\item[Influence maximization.]
The most well-studied problem related to influence spread is influence maximization~\cite{kempe03im}, that is, the problem of choosing a seed set $S$ to maximize the influence spread.
This has attracted much attention over the past two decades~\cite{li18survey,ling23deep,singh22survey}.
Although the influence maximization problem under most diffusion models including the IC model is NP-hard~\cite{kempe03im}, a simple greedy algorithm~\cite{kempe03im} can achieve a $(1-1/e)$-approximation provided that every $\sigma(S)$ is computed exactly.
Various methods, e.g., the sketch-based methods described in a survey~\cite{li18survey}, have improved this greedy algorithm.
However, if we resorted to using Monte Carlo simulation to evaluate $\sigma(S)$, we could not obtain a deterministic approximation bound for influence maximization, even with greedy or more sophisticated algorithms.
Moreover, although many studies have used an alternative proxy for $\sigma(S)$ to speed up the computation (e.g., proxy-based methods~\cite{li18survey} and learning-based methods~\cite{ling23deep}), they also lack approximation bounds.
This situation accentuates the theoretical importance of an exact evaluation of influence spread.

\item[Network reliability.]
Computing the probability that some vertices are connected in an uncertain graph has been traditionally studied as \emph{network reliability} evaluation in the network community.
Most studies on network reliability evaluation have modeled the network as an \emph{undirected} graph and, given a vertex set $K$ and the probability $p_e$ of presence for every edge, they computed the probability that the vertices in $K$ are connected.
Even for undirected graphs, computing this probability is \#P-complete~\cite{valiant79}.
Nevertheless, Hardy et al.~\cite{hardy2007knr} developed an algorithm that could compute this probability in linear time for graphs with bounded pathwidth.
Another line of research has attempted to compute the expected number of vertices connected to vertex $u$~\cite{nakamura23a}, which is an analogue of the influence spread under the IC model for undirected graphs.
Nakamura et al.~\cite{nakamura21} proposed an algorithm that could compute the probability that a vertex $u$ is connected to $v$ for \emph{every} $v\in V$ in linear time for an undirected graph with bounded pathwidth.
By summing all the probabilities for every $v$, it is also possible to compute the expected number of vertices connected to $v$ in linear time.
This algorithm uses the similarities in the computation of the probability that $u$ and $v$ are connected with Hardy's algorithm~\cite{hardy2007knr} for different $v$; therefore, it shares similar ideas with our work.
However, we cannot straightforwardly extend this algorithm to handle directed graphs, as described in Section~\ref{sec:intro}; detailed discussions are given in Section~\ref{sec:discussion}.
Note that these algorithms~\cite{hardy2007knr,nakamura21} also run in polynomial time for graphs with bounded pathwidth, indicating the importance of algorithms dealing with bounded-pathwidth graphs.
\end{description}

\section{Preliminaries}
\label{sec:preliminaries}
For sets $U,W$ of vertices, $U\reachto_G W$ means that $u$ can reach $w$, i.e., $w$ can be reached from $u$, for some $u\in U$ and $w\in W$ in directed graph $G$.
If $U$ (or $W$) consists of only a single vertex $u$ ($w$), we simply write, e.g., $u\reachto_G W$ and $u\reachto_G w$ instead of $\{u\}\reachto_G W$ and $\{u\}\reachto_G\{w\}$.
When it is clear from the context, we omit the subscript $G$.

A directed graph is called \emph{strongly connected} if $u\reachto v$ for every ordered pair $(u,v)$ of vertices.
We say vertex subset $C$ is a \emph{strongly connected component} (SCC) of $G$ if the vertex-induced subgraph of $G$ induced by $C$ is strongly connected but that induced by $V\cup\{v\}$ is not strongly connected for any $v\in V\setminus C$.
It is known that, given directed graph $G$, all SCCs in $G$ can be computed in linear time in the size of $G$~\cite{tarjan72SCC}.
Throughout this paper, $n$ and $m$ denote the number of vertices and edges, respectively, in the input graph $G$.

We formally define the problems to solve as follows.
\begin{problem}
  \label{prob:influence}
  Given a directed graph $G$, probabilities $\{p_e\}_{e\in E}$, and a seed set $S\subseteq V$, compute the probability $\prob{S\reachto v}$ for every $v\in V\setminus S$.
\end{problem}

For convenience, we assume that $G$ contains no self-loops and the underlying undirected graph of $G$ is connected.
Note that if $G$ contains self-loops, we can safely remove it, since it would never affect reachability.
If $G$ is disconnected, we can solve the problem individually for each connected component.

\section{Review of Previous Algorithm}
\label{sec:past}

\subsection{Decomposition of Probability}
In the following, we review Maehara's algorithm~\cite{maehara17} for exactly computing $\prob{S\reachto v}$ for a given $S,v$.
Although the literature \cite{maehara17} only showed the algorithm for the case $|S|=1$, here we show the general $S$ case because it can be easily extended.

The basic idea in evaluating $\prob{S\reachto v}$ is to decompose $\prob{S\reachto v}$ into terms that can be easily computed.
For $E^\prime,E^{\prime\prime}\subseteq E$ such that $E^\prime\cap E^{\prime\prime}=\emptyset$, let $\prob{\cdot\mid E^\prime,\overline{E^{\prime\prime}}}$ be the conditional probability given that the edges in $E^\prime$ are present and those in $E^{\prime\prime}$ are absent.
From the case analysis of whether $e\in E\setminus (E^\prime\cup E^{\prime\prime})$ is present or absent,
\begin{equation}
  \prob{S\reachto v\mid E^\prime,\overline{E^{\prime\prime}}}=p_e\cdot\prob{S\reachto v\mid E^\prime\cup\{e\},\overline{E^{\prime\prime}}}+(1-p_e)\cdot\prob{S\reachto v\mid E^\prime,\overline{E^{\prime\prime}\cup\{e\}}}. \label{eq:singledecomp}
\end{equation}
By ordering edges as $e_1,\ldots,e_m$, $\prob{S\reachto v}$ can be decomposed by recursively applying (\ref{eq:singledecomp}):
\begin{align}
  \prob{S\reachto v} & = p_{e_1}\cdot\prob{S\reachto v\mid \{e_1\},\overline{\emptyset}}+(1\!-\!p_{e_1})\cdot\prob{S\reachto v\mid \emptyset,\overline{\{e_1\}}} \nonumber \\
  & = p_{e_1}p_{e_2}\cdot\prob{S\reachto v\mid \{e_1,e_2\},\overline{\emptyset}}+p_{e_1}(1\!-\!p_{e_2})\cdot\prob{S\reachto v\mid \{e_1\},\overline{\{e_2\}}}+ \nonumber \\
  &\ \ (1\!-\!p_{e_1})p_{e_2}\cdot\prob{S\reachto v\mid \{e_2\},\overline{\{e_1\}}}+(1\!-\!p_{e_1})(1\!-\!p_{e_2})\cdot\prob{S\reachto v\mid \emptyset,\overline{\{e_1,e_2\}}} \nonumber \\
  & = \cdots = \sum_{E^\prime\subseteq E_{<i}}\probi{E^\prime}\cdot\prob{S\reachto v\mid E^\prime,\overline{E_{<i}\setminus E^\prime}}
  = \cdots, \label{eq:expansion1}\\
  \text{where}\quad & E_{<i}\coloneqq\{e_1,\ldots,e_{i-1}\},\quad\probi{E^\prime}\coloneqq\prod_{e\in E^\prime}p_e\cdot\prod_{e\in E_{<i}\setminus E^\prime}(1-p_e).\nonumber
\end{align}
This expansion eventually reaches the definition $\prob{S\reachto v}\coloneqq\sum_{E^\prime\in\mathcal{E}_{S\reachto v}}\prob{E^\prime}$.
Therefore, by using this decomposition naively, we have $\order{2^m}$ terms, which incurs exponential complexity of the graph size in evaluating it.

To prevent this, we attempt to detect the equality among probabilities, i.e., to find subsets $E^\prime,F^\prime\subseteq E_{<i}\ (E^\prime\neq F^\prime)$ such that $\prob{S\reachto v\mid E^\prime,\overline{E_{<i}\setminus E^\prime}}=\prob{S\reachto v\mid F^\prime,\overline{E_{<i}\setminus F^\prime}}$.
If such equality can be detected, we simply need to further decompose only one among the equal probabilities, which reduces the number of terms.
If the following condition (\#) holds, we can confirm $\prob{S\reachto v\mid E^\prime,\overline{E_{<i}\setminus E^\prime}}=\prob{S\reachto v\mid F^\prime,\overline{E_{<i}\setminus F^\prime}}$:
\begin{quote}
  (\#) For any $H\subseteq E_{\geq i}\coloneqq E\setminus E_{<i}$, $S\reachto_{(V,E^\prime\cup H)}v$ if and only if $S\reachto_{(V,F^\prime\cup H)}v$.    
\end{quote}

To check condition (\#), we focus on the reachability relation on the subgraph $(V,E^\prime)$.
If $u\reachto_{(V,E^\prime)} w$ and $w\reachto_{(V,F^\prime)} v$ are equivalent for any $u,w\in V$, we can confirm that (\#) holds.
In brief, this is because, given a path from $s\in S$ to $v$ on $(V,E^\prime\cup H)$, we can construct a path from $s$ to $v$ on $(V,F^\prime\cup H)$ by using the above equivalence, and vice versa.

Maehara et al.~\cite{maehara17} indicated that only a limited subset of vertices called \emph{frontier vertices} is important for checking (\#).
\begin{definition}
  \label{def:frontier}
  Given edge ordering $e_1,\ldots,e_m$, ($i$-th) \emph{frontier vertices} $W_i$ are those appearing in both $E_{<i}$ and $E_{\geq i}$.
\end{definition}
We consider the reachability relations among $W_i\cup S\cup\{v\}$ on $(V,E^\prime)$ as follows.
\begin{definition}
  Let $W_i^{S\reachto}(E^\prime)\coloneqq\{w\in W_i\mid S\reachto_{(V,E^\prime)}w\}$ (frontier vertices reachable from $S$) and $W_i^{\reachto v}(E^\prime)\coloneqq\{w\in W_i\mid w\reachto_{(V,E^\prime)}v\}$ (those that can reach $v$).
  A binary matrix $\Phi_i^v(E^\prime)$ is defined as follows: (I) It is indexed by $(W_i\setminus W_i^{S\reachto}(E^\prime)\cup\{S\})\times (W_i\setminus W_i^{\reachto v}(E^\prime)\cup\{v\})$. (II) $\Phi_i^v(E^\prime)_{uw}=1$ if and only if $u\reachto_{(V,E^\prime)}w$.
  Here, $S$ is treated as a special (single) row index, and $v$ is treated as a column index.
  We call $\Phi_i^v(E^\prime)$ a \emph{transversal configuration} (TC).
\end{definition}
We can then prove the following.
\begin{lemma}
  $\Phi_i^v(E^\prime)=\Phi_i^v(F^\prime)$ is a sufficient condition for (\#), meaning that $\prob{S\reachto v\mid E^\prime,\overline{E_{<i}\setminus E^\prime}}=\prob{S\reachto v\mid F^\prime,\overline{E_{<i}\setminus F^\prime}}$.
Note that $\Phi_i^v(E^\prime)=\Phi_i^v(F^\prime)$ means that the indices of matrices as well as their entries are equal.
\end{lemma}

This can be shown by transforming a $u$-$v$ directed path on $(V,E^\prime\cup H)$ for some $u\in S$ to a $u^\prime$-$v$ directed path on $(V,F^\prime\cup H)$ for some $u^\prime\in S$.
Now we can rewrite $\prob{S\reachto v\mid E^\prime,\overline{E_{<i}\setminus E^\prime}}$ as $\prob{S\reachto v\mid \Phi_i^v(E^\prime)}$.
\begin{example}
  On the graph in Figure~\ref{fig:example}a with $v=5$, we examine $\{e_1,e_4\},\{e_1,e_2,e_3,e_4\},$ $\{e_1,e_2,e_4,e_6\},\{e_2,e_5,e_6\}\subseteq E_{<8}$ corresponding to the four edge subgraphs of Figure~\ref{fig:example}b.
  These subsets have an identical TC $\Phi_8^5$ because $W_8=\{4,5\}$.
  Thus, the conditional probabilities of $S\reachto 5$ are all equal for these subsets.
\end{example}

\begin{figure}[tbp]
\begin{minipage}{0.56\columnwidth}
\begin{algorithm}[H]
{\footnotesize
  $\mathcal{N}_1^\prime\leftarrow\{\Phi_1^v(\emptyset)\}$, $\mathcal{N}_i^\prime\leftarrow\{\}\ (i\geq 2)$\;
  \For(\tcp*[f]{Diagram construction}){$i \leftarrow 1\ \KwTo\ m$}{
    \ForEach{$\Phi\in\mathcal{N}_i^\prime$}{
      \ForEach{$\bddf\in\{\bddlo,\bddhi\}$}{
        \lIf{$\bddf^v(\Phi)=\bot$}{$\Phi.\bddf\leftarrow\bot$}
        \lElseIf{$\bddf^v(\Phi)=\top$}{$\Phi.\bddf\leftarrow\top$}
        \lElse{ $\mathcal{N}_{i+1}^\prime\!\leftarrow\!\mathcal{N}_{i+1}^\prime\!\cup\!\{\bddf^v(\Phi)\}$, $\Phi.\bddf\!\leftarrow\!\bddf^v(\Phi)\!\in\!\mathcal{N}_{i+1}^\prime$}
      }
    }
  }
  $\Pmap[\Phi]\leftarrow1$ for $\Phi\in\mathcal{N}_1^\prime$, $\Pmap[\Phi]=0$ for $\Phi\in\mathcal{N}_i^\prime$ $(i\geq 2)$\;
  \For(\tcp*[f]{Top-down DP}){$i \leftarrow 1\ \KwTo\ m$}{
    \ForEach(\tcp*[f]{$\Pmap[\Phi]$ stores $\prob{\Phi}$}){$\Phi\in\mathcal{N}_i^\prime$}{
      $\Pmap[\Phi.\bddlo]\pluseq (1-p_{e_i})\cdot\Pmap[\Phi]$\;
      $\Pmap[\Phi.\bddhi]\pluseq p_{e_i}\cdot\Pmap[\Phi]$\;
    }
  }
  \Return $\Pmap[\top]$ \tcp*{$\prob{S\reachto v}=\prob{\top}$}
  \caption{Previous algorithm~\cite{maehara17}}
  \label{alg:previous}
}
\end{algorithm}
\end{minipage}
\begin{minipage}{0.43\columnwidth}
\centering
\includegraphics[keepaspectratio, scale=1.2]{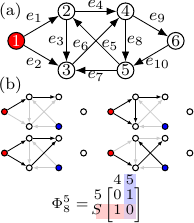}
\caption{(a) Example of input graph $G$. Red vertex is included in $S$. (b) Example of different edge subsets of $E_{<8}$ leading to the same TC $\Phi_8^5$.}\label{fig:example}
\end{minipage}
\end{figure}

Moreover, $\Phi_i^v(\cdot)$ has the following property: if $\Phi_i^v(E^\prime)=\Phi_i^v(F^\prime)$ for $E^\prime,F^\prime\subseteq E_{<i}$, both $\Phi_{i+1}^v(E^\prime)=\Phi_{i+1}^v(F^\prime)$ and $\Phi_{i+1}^v(E^\prime\cup\{e_i\})=\Phi_{i+1}^v(F^\prime\cup\{e_i\})$ hold.
Thus, there exist transition functions $\bddlo^v,\bddhi^v$ from $\Phi_i^v$ to $\Phi_{i+1}^v$ satisfying $\bddlo^v(\Phi_i^v(E^\prime))=\Phi_{i+1}^v(E^\prime)$ and $\bddhi^v(\Phi_i^v(E^\prime))=\Phi_{i+1}^v(E^\prime\cup\{e_i\})$ for any $E^\prime\subseteq E_{<i}$.
By combining this fact and (\ref{eq:singledecomp}), we have, for any $\Phi=\Phi_i^v(E^\prime)\ (E^\prime\subseteq E_{<i})$,
\begin{equation}
  \prob{S\reachto v\mid \Phi}=p_e\cdot\prob{S\reachto v\mid\bddhi^v(\Phi)}+(1-p_e)\cdot\prob{S\reachto v\mid \bddlo^v(\Phi)}. \label{eq:phidecomp}
\end{equation}
Using this decomposition recursively, $\prob{S\reachto v}$ can be decomposed as follows:
\begin{align}
  \prob{S\reachto v} & = \prob{S\reachto v\mid\Phi_1^v(\emptyset)} \nonumber \\
  & = p_{e_1}\cdot\prob{S\reachto v\mid\bddhi^v(\Phi_1^v(\emptyset))}+(1\!-\!p_{e_1})\cdot\prob{S\reachto v\mid\bddlo^v(\Phi_1^v(\emptyset))} \nonumber \\
  & = \cdots = \sum_{\Phi\in\mathcal{N}_i}\prob{\Phi}\cdot\prob{S\reachto v\mid\Phi} = \cdots, \label{eq:expansion2} \\
  \text{where}\quad & \prob{\Phi} \coloneqq\ \ \smashoperator[l]{\sum_{E^\prime\subseteq E_{<i}:\Phi_i^v(E^\prime)=\Phi}}\probi{E^\prime},\quad \mathcal{N}_i\coloneqq\{\Phi_i^v(E^\prime)\mid E^\prime\subseteq E_{<i}\}. \label{eq:TCprob}
\end{align}
In other words, $\prob{\Phi}$ is the probability that TC $\Phi$ appears and $\mathcal{N}_i$ is the set of all TCs of the subsets of $E_{<i}$.
We later show that $|\mathcal{N}_i|$ is bounded by a constant if the pathwidth is bounded in Lemma~\ref{lem:pattern_phi}.
This is the main reason to achieve $\order{m}$ time complexity in computing $\prob{S\reachto v}$.

For some TCs $\Phi$, we can confirm that $S\reachto v$ must hold or must not hold, i.e., $\prob{S\reachto v\mid\Phi}=1$ or $0$.
For such TCs, we do not need to further decompose $\prob{S\reachto v\mid\Phi}$ with (\ref{eq:phidecomp}).
We define such base cases with the following $\top$-pruning and $\bot$-pruning as follows.
Let $i_u$ be the largest index $i$ such that one of the endpoints of $e_{i}$ is $u$ and let $i_S\coloneqq\max_{u\in S}i_u$.
\begin{itemize}
  \item If $\bddhi^v(\Phi)_{Sv}=1$, determining the presence of $e_i$ confirms $S\reachto v$, i.e., $\prob{S\reachto v\mid\bddhi^v(\Phi)}=1$.
  In such a case, we let $\bddhi^v(\Phi)=\top$, a special TC meaning that $S\reachto v$ must hold ($\top$-pruning).
  \item Suppose that $i\geq i_S$ and $\bddlo^v(\Phi)_{Su}=0$ for any $u$, or that $i\geq i_v$ and $\bddlo^v(\Phi)_{uv}=0$ for any $u$.
  In these cases, no further vertices can be reached from $S$ or reach $v$ by determining the absence of $e_i$, i.e., $\prob{S\reachto v\mid\bddlo^v(\Phi)}=0$.
  Thus, we let $\bddlo^v(\Phi)=\bot$, which is a special TC meaning that $S\reachto v$ must not hold ($\bot$-pruning)\footnote{Although the original version~\cite{maehara17} uses different pruning rules, they require a pre-computation of transitive closures of $(V,E_{\geq i})$, which takes cubic time. We use an alternative pruning rule that can be checked in linear time.}.
  We also let $\bddhi^v(\Phi)=\bot$ if $\top$-pruning does not occur and $\bddhi^v(\Phi)$ satisfies the same condition as above.
\end{itemize}

Eventually, all the TCs are transformed into either $\top$ or $\bot$, and (\ref{eq:expansion2}) reaches $\prob{S\reachto v}=\prob{\top}\cdot 1+\prob{\bot}\cdot 0=\prob{\top}$.
Thus, computing $\prob{\top}$ is sufficient for evaluating $\prob{S\reachto v}$.

\subsection{Procedure and Complexity}
Now we describe the procedure to compute $\prob{\top}$.
We again use the transition of TCs.
By considering the pruning described above, we redefine the set of TCs as $\mathcal{N}^\prime_1\coloneqq\{\Phi_1^v(\emptyset)\}$ and $\mathcal{N}^\prime_i\coloneqq\left(\bigcup_{\Phi\in\mathcal{N}_{i-1}^\prime}\{\bddhi^v(\Phi),\bddlo^v(\Phi)\}\right)\setminus\{\top,\bot\}$, i.e., we remove the pruned TCs from $\mathcal{N}_i$.
Then, from the definition of $\prob{\Phi}$ (Eq. (\ref{eq:TCprob})), we have
\begin{align}
  \prob{\Phi}=p_i\cdot\sum_{\Phi^\prime\in\mathcal{N}^\prime_{i-1}:\bddhi(\Phi^\prime)=\Phi}\prob{\Phi^\prime}+(1-p_i)\cdot\sum_{\Phi^\prime\in\mathcal{N}^\prime_{i-1}:\bddlo(\Phi^\prime)=\Phi}\prob{\Phi^\prime}. \label{eq:decompphi}
\end{align}
To compute $\prob{\top}$, we can use a similar formula:
\begin{align}
  \prob{\top}=p_i\cdot\sum_{\Phi^\prime\in\mathcal{N}^\prime_{1}\cup\cdots\cup\mathcal{N}^\prime_m:\bddhi(\Phi^\prime)=\top}\prob{\Phi^\prime}+(1-p_i)\cdot\sum_{\Phi^\prime\in\mathcal{N}^\prime_{1}\cup\cdots\cup\mathcal{N}^\prime_m:\bddlo(\Phi^\prime)=\top}\prob{\Phi^\prime}. \label{eq:decomptop}
\end{align}

Starting from $\mathcal{N}_1^\prime=\{\Phi_1^v(\emptyset)\}$, we can generate $\mathcal{N}_{i+1}^\prime$ from $\mathcal{N}_i^\prime$ recursively using transition functions $\bddhi^v,\bddlo^v$.
By recording the transition, we can compute every $\prob{\Phi}$ and $\prob{\top}$ with (\ref{eq:decompphi}) and (\ref{eq:decomptop}).

Algorithm~\ref{alg:previous} formally presents the procedures.
Lines 1--7 recursively generate the successive TCs using transition functions $\bddlo^v,\bddhi^v$ while executing $\bot,\top$-pruning.
This process can be seen as constructing a diagram of TCs as in Figure~\ref{fig:diagram}.
Every $\Phi$ except for $\bot,\top$ has two outgoing edges heading to $\bddlo^v(\Phi)$ and $\bddhi^v(\Phi)$.
In Algorithm~\ref{alg:previous}, each $\Phi$ has two entries $\Phi.\bddlo$ and $\Phi.\bddhi$ that store the pointer to $\bddlo^v(\Phi)$ and $\bddhi^v(\Phi)$.
Here, $\mathcal{N}_i^\prime$ stores the TCs of the subsets of $E_{<i}$ as a hash map; pruned TCs are not stored.
Along this diagram, we can compute every $\prob{\Phi}$ with simple dynamic programming (DP) in lines 8--13; this simulates the formulas (\ref{eq:decompphi}) and (\ref{eq:decomptop}).
The array $\Pmap[\Phi]$ for $\Phi\in\mathcal{N}_i^\prime$ is prepared to store $\prob{\Phi}$.
Note that even if an identical TC $\Phi$ appears in $\mathcal{N}_i^\prime$ and $\mathcal{N}_j^\prime$ $(i\neq j)$, we distinguish $\Phi\in\mathcal{N}_i^\prime$ and $\Phi\in\mathcal{N}_j^\prime$ in computing $\Pmap[\cdot]$.
After computing all $\Pmap[\Phi]$, $\prob{S\reachto v}=\prob{\top}$ is stored in $\Pmap[\top]$.

Next, we analyze the complexity of this algorithm.
For each TC $\Phi$, computing $\bddlo^v(\Phi)$ and $\bddhi^v(\Phi)$ takes $\order{(|W_i|+|W_{i+1}|)^2}$ time; details are in Lemma~\ref{lem:transition} of Section~\ref{ssec:complexity}.
If every $\mathcal{N}_i^\prime$ is implemented as a hash map, storing or searching a TC also takes $\order{|W_i|^2}$ time.
The number of patterns on $\Phi_i^v$ can be bounded by $\order{2^{|W_i|^2}}$;
since this was not rigorously proved~\cite{maehara17}, we prove it in Lemma~\ref{lem:pattern_phi} of Section~\ref{ssec:complexity}.
The overall complexity can be bounded by $\order{\sum_i (|W_i|+|W_{i+1}|)^2\cdot 2^{|W_i|^2}}=\order{m\omega^2\cdot 2^{\omega^2}}$, where $\omega=\max_i|W_i|$.

In fact, $\omega$ is related to the pathwidth of a graph.
Given a path decomposition of $G$ with width $\omega_p$, we can generate an edge ordering with $\omega\leq \omega_p$ in linear time~\cite{inoue2016pathwidth}.
Since we can construct a path decomposition in linear time for a graph with bounded pathwidth~\cite{bodlaender96tw}, we have a linear time algorithm for computing $\prob{S\reachto v}$ on a graph with bounded pathwidth.
However, computing $\sigma(S)$ requires $\order{mn\omega^2\cdot 2^{\omega^2}}$ time by computing $\prob{S\reachto v}$ for $v\in V\setminus S$.

\begin{figure}[!tb]
    \centering
    \includegraphics[keepaspectratio,scale=1.05]{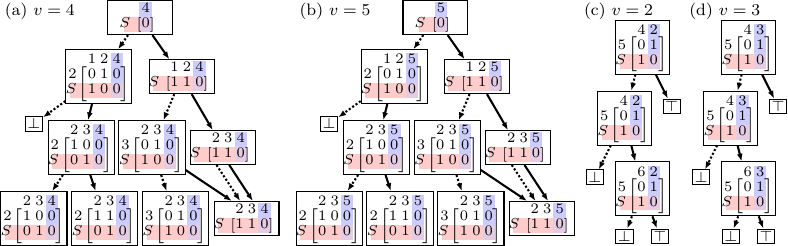}
    \caption{(a)(b) Upper four levels of diagrams constructed from graph in Figure~\ref{fig:example}a with $v=4$ and $v=5$. (c)(d) Parts of lower levels of diagrams constructed from graph in Figure~\ref{fig:example}a with $v=2$ and $v=3$; (c) depicts part below $\Phi_8^2(\{e_2,e_5,e_6\})$ and (d) depicts part below $\Phi_8^3(\{e_1,e_4,e_7\})$. Solid arcs indicate connection from $\Phi$ to $\Phi.\bddhi$, while dashed arcs indicate connection from $\Phi$ to $\Phi.\bddlo$.}
    \label{fig:diagram}
\end{figure}

\section{Linear-Time Influence Spread Computation}
\label{sec:influence}
We now describe the proposed algorithm for Problem~\ref{prob:influence}.
First, we explain the high-level concept of the proposed algorithm.
In computing $\sigma(S)$ with the algorithm described in Section~\ref{sec:past}, different diagrams of TCs are generated for different $v\in V\setminus S$ in computing $\prob{S\reachto v}$.
To argue the similarities among these diagrams, we divide each of them into three parts: \emph{upper levels}, \emph{lower levels}, and \emph{middle levels}.
More specifically, let $A_i\subseteq V$ be the vertices that only appear in $E_{<i}$ and do not appear in $E_{\geq i}$, and $B_i\subseteq V$ be those that only appear in $E_{\geq i}$ and do not appear in $E_{<i}$.
$A_i$, $B_i$, and $W_i$ (Definition~\ref{def:frontier}) constitute a partition of $V$.
Moreover, since $E_{<i}\subset E_{<i+1}$ and $E_{\geq i}\supset E_{\geq i+1}$ for $i=1,\ldots,m$, $V=B_1\supseteq \cdots\supseteq B_{m+1}=\emptyset$ and $\emptyset=A_1\subseteq\cdots\subseteq A_{m+1}=V$.
Thus, by fixing $v$, the diagram of TCs $\Phi_i^v(\cdot)$ made from the reachability to $v$ can be divided into three parts according to $i$: levels $i$ s.t. $v\in B_i$ (upper levels), those s.t. $v\in W_i$ (middle levels), and those s.t. $v\in A_i$ (lower levels).

We first find that upper levels of these diagrams exhibit an identical structure, i.e., diagrams of $\Phi_i^v$ and $\Phi_i^w$ exhibit an identical structure as long as $v,w\in B_i$ (Observation~\ref{obs:top}); we extract the essential information to build this identical diagram as a \emph{shared transversal configuration} (STC) $\Psi=\Psi_i(E^\prime)$.
We can also confirm that middle levels of these diagrams were also governed by the STCs (Lemma~\ref{lem:equiv2}).
Thus, by considering STCs $\Psi_i(\cdot)$, we can have a similar formula to (\ref{eq:expansion2}) using the STCs when $v\in B_i\cup W_i$:
\begin{align}
  & \prob{S\reachto v} = \sum_{\Psi\in\mathcal{M}_i}\prob{\Psi}\cdot\prob{S\reachto v\mid\Psi}, \label{eq:expansion3} \\
  & \text{where}\quad
  \prob{\Psi} \coloneqq\ \ \smashoperator[l]{\sum_{E^\prime\subseteq E_{<i}:\Psi_i(E^\prime)=\Psi}}\probi{E^\prime},\quad\mathcal{M}_i\coloneqq\{\Psi_i(E^\prime)\mid E^\prime\subseteq E_{<i}\}. \nonumber
\end{align}

For every $v$, there exists $i$ such that $v\in B_i\cup W_i$.
Thus, for every $v$, if $\prob{\Psi}$ and $\prob{S\reachto v\mid\Psi}$ are efficiently computed for all $\Psi=\Psi_i(\cdot)$ with the above $i$, we can compute $\prob{S\reachto v}$ with (\ref{eq:expansion3}).
Here, since $\prob{\Psi}$ has a similar recursion as $\prob{\Phi}$ (Lemma~\ref{lem:pdp}), all $\prob{\Psi}$s can be computed in a manner similar to that used in Section~\ref{sec:past} by building a diagram of STCs from $i=1$ to $m+1$.
Since $|\mathcal{M}_i|$ is constant for a graph with bounded pathwidth (Lemma~\ref{lem:pattern_psi}), it can be performed in $\order{m}$ time.

Indeed, when $v\in W_i$ (middle levels), $\prob{S\reachto v\mid\Psi}$ for $\Psi=\Psi_i(\cdot)$ can also be computed efficiently.
To compute this, we use the lower levels of the diagram of TCs $\Psi_i^v(\cdot)$.
Although naively constructing the diagrams of TCs $\Phi$ for every $v$ costs $\order{mn}$ time, we find that the lower levels of these diagrams have many identical substructures (Lemma~\ref{lem:equiv1}).
By using this, we can build all possible substructures of diagrams that appear in the lower level of these diagrams in $\order{m}$ time.
In addition, we again use middle levels' equivalence (Lemma~\ref{lem:equiv2}) to compute $\prob{S\reachto v\mid\Psi}$ on middle levels in $\order{m}$ time.
Here we use the recursion on $\prob{S\reachto v\mid\Psi}$ (Lemmas~\ref{lem:rdp1} and \ref{lem:rdp2}).

Eventually, for every $v$ and $i$ where $v\in W_i$, we obtain $\prob{\Psi}$ and $\prob{S\reachto v\mid\Psi}$ for $\Psi=\Psi_i(\cdot)$.
Thus, by choosing $i$ such that $v\in W_i$, we can compute $\prob{S\reachto v}$ with (\ref{eq:expansion3}).
Since (\ref{eq:expansion3}) contains $\order{|\mathcal{M}_i|}$ terms, which is constant for bounded-pathwidth graphs, the computation of $\prob{S\reachto v}$ for every $v$ costs $\order{n}$ time in total.
The overall time complexity is bounded by $\order{m+n}$ for a graph with bounded pathwidth.
Note that, for every $v$, there exists $i$ such that $v\in W_i$ if $v$ has degree more than $1$.
For vertices with degree $1$, such $i$ might not exist, but we can easily compute $\prob{S\reachto v}$ for such vertices, as described later.

In the following, we first identify the similarities in the previous algorithm~\cite{maehara17} for different $v$ in Section~\ref{ssec:equiv}.
Then, we derive DP formulas for computing $\prob{S\reachto v}$ in Section~\ref{ssec:sketch}.
Finally, we explain the main procedures and the complexity in Sections~\ref{ssec:procedures} and \ref{ssec:complexity}.

\subsection{Equivalence of Diagrams for Different Vertices}
\label{ssec:equiv}
We formally state the equivalence of different diagrams of TCs for upper levels ($v\in B_i$), lower levels ($v\in A_i$), and middle levels ($v\in W_i$).

\subparagraph*{Upper level equivalence.}
First, we consider the case $v,w\in B_i$.
No edges in $E_{<i}$ are incident to $v$, meaning that no vertex in $W_i$ can reach $v$ on $(V,E_{<i})$.
Furthermore, no vertex in $S$ can reach $v$ on $(V,E_{<i})$.
Therefore, for any $E^\prime\subseteq E_{<i}$, $\Phi_i^v(E^\prime)_{uv}=0$ for any $u$.
This means we need to consider the reachability relations among only $W_i\cup\{S\}$, not $W_i\cup\{S,v\}$.
\begin{definition}
  Let $\Psi_i(E^\prime)$ be a binary matrix indexed by $(W_i\setminus W_i^{S\reachto}(E^\prime)\cup\{S\})\times W_i$, where $\Psi_i(E^\prime)_{ux}=1$ if and only if $u\reachto_{(V,E^\prime)}x$.
  Here, $S$ is treated as a special (single) row index.
  We call $\Psi_i(\cdot)$ a \emph{shared transversal configuration} (STC).
\end{definition}

Since $v\in B_i$, we can recover $\Phi_i^v(E^\prime)$ from $\Psi_i(E^\prime)$ by $\Phi_i^v(E^\prime)_{ux}=\Psi_i(E^\prime)_{ux}$ for any $u$ as well as $x\neq v$ and $\Phi_i^v(E^\prime)_{uv}=0$ for any $u$.
In other words, the diagram structure is completely determined by the transition of $\Psi_i(\cdot)$ instead of $\Phi_i^v(\cdot)$ when $v\in B_i$.
The above argument also holds for the other $w\in B_i$.
Therefore, the following observation holds.
\begin{observation}
  \label{obs:top}
  Consider the diagrams generated using Algorithm~\ref{alg:previous} with different $v,w$.
  If $v,w\in B_i$, the diagrams above any TCs $\Phi_i^v(E^\prime)$ and $\Phi_i^w(E^\prime)$ are identical.
  Moreover, we can instead use STC $\Psi_i(\cdot)$ in building the diagram.
\end{observation}
\begin{example}
  Figures~\ref{fig:diagram}a and \ref{fig:diagram}b show the upper four levels of the diagram generated using Algorithm~\ref{alg:previous} with the graph of Figures~\ref{fig:example}a and $v=4$ or $v=5$.
  Since $4,5\in B_4$, the first four levels are identical, as depicted in these figures.
\end{example}

\subparagraph*{Lower level equivalence.}
Next, we consider the case $v,w\in A_i$.
Recall that $\Phi_i^v(\cdot)$ captures the reachability among $W_i\cup\{S,v\}$, and $\Phi_i^{w}(\cdot)$ captures those among $W_i\cup\{S,w\}$.
Since $v,w\in A_i$, all edges incident to $v$ and $w$ are in $E_{<i}$, meaning that all incident edges are determined when the edges in $E_{<i}$ are determined.
Thus, intuitively, if $\Phi_i^v$ and $\Phi_i^{w}$ are ``identical,'' i.e., the reachability relations are identical by substituting $v$ with $w$, $S\reachto v$ and $S\reachto w$ can be regarded as equivalent events.
We formally define such an equivalence.
\begin{definition}
For $E^\prime,F^\prime\subseteq E_{<i}$, consider two TCs $\Phi_i^v(E^\prime)$ and $\Phi_i^{w}(F^\prime)$, where $v,w\in A_i$.
We say these TCs are \emph{identical} if $\Phi_i^v(E^\prime)_{ux}=\Phi_i^{w}(F^\prime)_{ux}$ for any $u$ and $x\neq v,w$ and $\Phi_i^v(E^\prime)_{uv}=\Phi_i^{w}(F^\prime)_{uw}$ for any $u$.
This means that $\Phi_i^{v}(E^\prime)$ and $\Phi_i^{w}(F^\prime)$ are identical matrices, except for the column indices being different.
\end{definition}

\begin{restatable}{lemma}{Equivone}
  \label{lem:equiv1}
  Suppose $\Phi_i^v(E^\prime)$ and $\Phi_i^{w}(F^\prime)$ are identical for $v,w\in A_i$ and $E^\prime,F^\prime\subseteq E_{<i}$.
  Then, for any $H\subseteq E_{\geq i}$, $S\reachto_{(V,E^\prime\cup H)}v$ if and only if $S\reachto_{(V,F^\prime\cup H)}w$.
\end{restatable}
\begin{proof}
  Fix $H\subseteq E_{\geq i}$ and suppose that $S\reachto_{(V,E^\prime\cup H)}v$.
  There is then a directed walk $P$ from some $s\in S$ to $v$ in $(V,E^\prime\cup H)$.
  The walk $P$ can be divided into small walks $P=P_0P_1P_2\cdots P_{2k-1}P_{2k}$ satisfying the conditions that $P_{2j}$ consists of only the edges in $E^\prime\subseteq E_{<i}$ and that $P_{2j+1}$ consists of only the edges in $H\subseteq E_{\geq i}$.
  Let $p_j$ be the contact vertex of $P_j$ and $P_{j+1}$.
  Note that $P_0$ may be empty when $s\in W_i\cup B_i$; in such a case, we set $p_0=s$.
  Since $p_j\ (j\geq 1)$ is incident to both edges in $E^\prime\subseteq E_{<i}$ and those in $H\subseteq E_{\geq i}$, $p_j\in W_i$.
  If $P_0$ is nonempty, $p_0$ is also in $W_i$.
  We can also assume that $p_1,p_2,\ldots,p_{2k-2}\notin W_i^{S\reachto}(E^\prime)$; otherwise, we can take a directed walk $P^\star$ from some $s\in S$ to $p_{j}$ within $(V,E^\prime)$ and consider a walk $P^\star P_{j+1}\cdots P_{2k}$ instead.
  With a similar argument, we can assume that $p_1,p_2,\ldots,p_{2k-2}\notin W_i^{\reachto v}(E^\prime)$.
  Now, we constitute a directed walk $P^\prime$ from some $s^\prime\in S$ to $w$ in $(V,F^\prime\cup H)$.
  First, when $P_0$ is not empty, $\Phi_i^v(E^\prime)_{Sp_0}=1$.
  Since $\Phi_i^v(E^\prime)$ and $\Phi_i^{w}(F^\prime)$ are identical, $\Phi_i^{w}(F^\prime)_{Sp_0}=1$; thus, there is a directed walk $P^\prime_0$ from some $s^\prime\in S$ to $p_0$ in $(V,F^\prime)$.
  This also holds when $P_0$ is empty; in such a case, $s^\prime=s=p_0$ and $P^\prime_0$ is also empty.
  Next, for $1\leq j\leq k-1$, $P_{2j}$'s existence and identity indicate $\Phi_i^v(E^\prime)_{p_{2j-1}p_{2j}}=1=\Phi_i^w(F^\prime)_{p_{2j-1}p_{2j}}$, meaning that there is a directed walk $P^\prime_{2j}$ from $p_{2j-1}$ to $p_{2j}$ in $(V,F^\prime)$.
  Finally, $P_{2k}$'s existence and identity indicate $\Phi_i^v(E^\prime)_{p_{2k-1}v}=1=\Phi_i^w(F^\prime)_{p_{2k-1}w}$, meaning that there is a directed walk $P_{2k}^\prime$ from $p_{2k-1}$ to $w$ in $(V,F^\prime)$.
  Now $P^\prime=P_0^\prime P_1 P_2^\prime\cdots P_{2k-1}P_{2k}^\prime$ is a directed walk from $s^\prime\in S$ to $w$ in $(V,F^\prime\cup H)$, indicating $S\reachto_{(V,F^\prime\cup H)} w$.
  
  In the same manner, we can show the reverse direction.
  Thus, the lemma holds.
\end{proof}
This indicates that $\prob{S\reachto v\mid\Phi_i^v}=\prob{S\reachto w\mid\Phi_i^w}$ when $\Phi_i^v$ and $\Phi_i^w$ are identical.

\begin{observation}
  \label{obs:bottom}
  Consider the diagrams generated using Algorithm~\ref{alg:previous} with different $v,w$.
  If $v,w\in A_i$ and identical TCs $\Phi_i^v,\Phi_i^w$ appear in each diagram, the diagram structures below $\Phi_i^v$ and below $\Phi_i^w$ are identical.
\end{observation}
\begin{example}
  Figures~\ref{fig:diagram}c and \ref{fig:diagram}d show parts of the last four levels of the diagrams below $\Phi_8^2(\{e_2,e_5,e_6\})$ $(v=2)$ and below $\Phi_8^3(\{e_1,e_4,e_7\})$ $(v=3)$.
  Since $2,3\in A_8$ and these TCs are identical, these diagram structures are also identical.
\end{example}

\subparagraph*{Middle level equivalence.}
Finally, we examine the case $v,w\in W_i$.
In this case, the TC $\Phi_i^v(\cdot)$ captures the reachability among $W_i\cup\{S,v\}=W_i\cup\{S\}$.
In addition, the STC $\Psi_i(\cdot)$ also captures the reachability among $W_i\cup\{S\}$.
Thus, we can recover $\Phi_i^v(E^\prime)$ from $\Psi_i(E^\prime)$ by removing the columns corresponding to $W_i^{\reachto v}(E^\prime)$.
In other words, if $\Psi_i(E^\prime)=\Psi_i(F^\prime)$, then $E^\prime,F^\prime\subseteq E_{<i}$ satisfy condition (\#).

For $\Psi=\Psi_i(E^\prime)$, let $W_i^{S\reachto}(\Psi)$ be the set of vertices in $W_i$ that can be reached from $S$ under $\Psi$, which is determined by the missing row indices of $\Psi$.
Consider a graph $G_i(\Psi)\coloneqq(W_i\setminus W_i^{S\reachto}(\Psi),E_i(\Psi))$, where $E_i(\Psi)\coloneqq\{(u,x)\mid u,x\in W_i\setminus W_i^{S\reachto}(\Psi),\Psi_{ux}=1\}$.
In other words, when $\Psi=\Psi_i(E^\prime)$, $G_i(\Psi)$ is the subgraph of the transitive closure of $(V,E^\prime)$ induced by $W_i\setminus W_i^{S\reachto}(\Psi)$.
Then, we can prove the following.
\begin{restatable}{lemma}{Equivtwo}
  \label{lem:equiv2}
  Suppose $\Psi_i(E^\prime)=\Psi_i(F^\prime)\eqqcolon\Psi$ for $E^\prime,F^\prime\subseteq E_{<i}$ and $v,w\in W_i\setminus W_i^{S\reachto}(E^\prime)$ are in the same SCC of $G_i(\Psi)$.
  Then, for any $H\subseteq E_{\geq i}$, $S\reachto_{(V,E^\prime\cup H)}v$ if and only if $S\reachto_{(V,F^\prime\cup H)}w$.
\end{restatable}
\begin{proof}
  Fix $H\subseteq E_{\geq i}$.
  We can prove the equivalence of $S\reachto_{(V,E^\prime\cup H)}v$ and $S\reachto_{(V,F^\prime\cup H)}v$ from $\Psi_i(E^\prime)=\Psi_i(F^\prime)$ with almost the same argument as that in the proof of Lemma~\ref{lem:equiv1}.
  Since $v$ and $w$ are in the same SCC of $G_i(\Psi)$, $v\reachto_{(V,F^\prime)}w$ and $w\reachto_{(V,F^\prime)}v$.
  Consequently, $S\reachto_{(V,F^\prime\cup H)}v$ indicates $S\reachto_{(V,F^\prime\cup H)}w$ and vice versa.
  This concludes $S\reachto_{(V,E^\prime\cup H)}v\iff S\reachto_{(V,F^\prime\cup H)}w$.
\end{proof}

\subsection{New Decomposition and DP Formulas}
\label{ssec:sketch}
\subparagraph*{Decomposition formula.}
Using the equivalences described above, we derive a formula for $\prob{S\reachto v}$ using the STCs.
Since we confirmed that $\Phi_i^v(E^\prime)$ can be recovered from $\Psi_i(E^\prime)$ when $v\in A_i\cup W_i$, we can derive the decomposition formula (\ref{eq:expansion3}).
If $v\in W_i$, it can be further transformed.
Based on Lemma~\ref{lem:equiv2}, we can rewrite $\prob{S\reachto v\mid\Psi}$ as $\prob{S\reachto C_\Psi(v)\mid\Psi}$, where $C_\Psi(v)$ is the SCC of $G_i(\Psi)$ containing $v$ when $v\notin W_i^{S\reachto}(\Psi)$.
If $v\in W_i^{S\reachto}(\Psi)$, we immediately have $\prob{S\reachto v\mid\Psi}=1$.
Now (\ref{eq:expansion3}) can be further transformed into
\begin{align}
  \prob{S\reachto v} & = \sum_{\Psi\in\mathcal{M}_i}
  \begin{cases}
    \prob{\Psi} & (v\in W_i^{S\reachto}(\Psi)) \\
    \prob{\Psi}\cdot\prob{S\reachto C_\Psi(v)\mid\Psi} & (v\in W_i\setminus W_i^{S\reachto}(\Psi))
  \end{cases}
  . \label{eq:expansion4}
\end{align}
If we can obtain $\prob{\Psi}$ and $\prob{S\reachto C\mid\Psi}$ for every STC $\Psi$ and every SCC $C$ of $G_i(\Psi)$, we can compute every $\prob{S\reachto v}$ by choosing $i$ such that $v\in W_i$ and then use (\ref{eq:expansion4}).
Note that if all vertices have degree not less than $2$, there must exist index $i$ such that $v\in W_i$, since $G$ contains no self-loops.
The $\prob{S\reachto v}$ value for degree $1$ vertices can be easily obtained; details are given in Section~\ref{ssec:procedures}.

\subparagraph*{DP formulas.}
Since TCs can be recovered from STCs and admit transition functions $\bddhi^v,\bddlo^v$, STCs also admit transition functions $\bddlo$ and $\bddhi$ satisfying that, for any $E^\prime\subseteq E_{<i}$, $\bddlo(\Psi_i(E^\prime))=\Psi_{i+1}(E^\prime)$ and $\bddhi(\Psi_i(E^\prime))=\Psi_{i+1}(E^\prime\cup\{e_i\})$.
Here, we set $\bddf(\Psi)=\bot$ for $\bddf\in\{\bddhi,\bddlo\}$ when $i\geq i_S$ and $\bddf(\Psi)_{Su}=0$ for any $u$, which confirms $\prob{S\reachto w\mid\bddf(\Psi)}=0$ for any further vertices $w\in W_{i+1}\cup B_{i+1}$ ($\bot$-pruning).
Now $\prob{\Psi}$ can be computed in the same manner as described in Section~\ref{sec:past}:
From the definition of $\prob{\Psi}$ (Eq. (\ref{eq:expansion3})), we have the following.
\begin{lemma}
  \label{lem:pdp}
  For $\Psi=\Psi_i(E^\prime)$ for some $E^\prime\subseteq E_{<i}$, we have
  \begin{align}
    \prob{\Psi}=p_i\cdot\sum_{\Psi^\prime\in\mathcal{M}^\prime_{i-1}:\bddhi(\Psi^\prime)=\Psi}\prob{\Psi^\prime}+(1-p_i)\cdot\sum_{\Psi^\prime\in\mathcal{M}^\prime_{i-1}:\bddlo(\Psi^\prime)=\Psi}\prob{\Psi^\prime}, \label{eq:pdp}
  \end{align}
  where $\mathcal{M}^\prime_1\coloneqq\{\Psi_1(\emptyset)\}$ and $\mathcal{M}^\prime_i\coloneqq\left(\bigcup_{\Psi\in\mathcal{M}_{i-1}^\prime}\{\bddhi(\Psi),\bddlo(\Psi)\}\right)\setminus\{\bot\}$, i.e., we remove pruned STCs from $\mathcal{M}_i$.
\end{lemma}

It remains to be explained how to compute $\prob{S\reachto C\mid\Psi}$ for an SCC $C$ of $G_i(\Psi)$.
If there exists $u\in C$ such that $u\in W_{i+1}$, we can decompose it using the case analysis of $e_i$:
If $e_i$ is present, $\prob{S\reachto C\mid\Psi}$ equals $\prob{S\reachto C^\prime\mid\bddhi(\Psi)}$, where $C^\prime$ is the SCC of $G_{i+1}(\bddhi(\Psi))$ containing $u$.
If $e_i$ is absent, $\prob{S\reachto C\mid\Psi}$ equals $\prob{S\reachto C^{\prime\prime}\mid\bddlo(\Psi)}$, where $C^{\prime\prime}$ is the SCC of $G_{i+1}(\bddlo(\Psi))$ containing $u$.
Here, the relation between $C$ and $C^\prime,C^{\prime\prime}$ can be seen as the transition from $C$ to successive SCCs $C^\prime,C^{\prime\prime}$.
By considering the pruned cases, we formally define the successive SCCs as follows.
\begin{definition}
  \label{def:successive}
  For SCC $C$ of $G_i(\Psi)$, we define \emph{successive SCC} $\bddf(C)$ for $\bddf\in\{\bddlo,\bddhi\}$ as follows.
  (i) If $\bddf=\bddhi$, and $e_i=(u,v)$ with $u\in W_i^{S\reachto}(\Psi)\cup S$ and $v\in C$, $\bddhi(C)\coloneqq\top$, meaning that $C$ can be reached from $S$ by determining the presence of $e_i$.
  (ii) If (i) does not hold and $\bddf(\Psi)=\bot$, $\bddf(C)\coloneqq\bot$, meaning that $C$ cannot be reached from $S$.
  (iii) If neither (i) nor (ii) holds and there exists $u\in C$ such that $u\in W_{i+1}$, $\bddf(C)\coloneqq C_{\bddf(\Psi)}(u)$. 
  (iv) Otherwise, $\bddf(C)\coloneqq\emptyset$, meaning that there is no successive SCC.
\end{definition}
\begin{lemma}
  \label{lem:rdp1}
  If there exists $u\in C$ such that $u\in W_{i+1}$, i.e., successive SCCs exist, we can decompose $\prob{S\reachto C\mid \Psi}$ for $\Psi=\Psi_i(E^\prime)$ and SCC $C$ of $G_i(\Psi)$ as follows:
  \begin{align}
    \prob{S\reachto C\mid\Psi}=p_{e_i}\cdot\prob{S\reachto\bddhi(C)\mid \bddhi(\Psi)}+(1-p_{e_i})\cdot\prob{S\reachto\bddlo(C)\mid \bddlo(\Psi)}. \label{eq:rdp1}
  \end{align}
\end{lemma}
Note that if $u\in W_{i+1}^{S\reachto}(\bddhi(\Psi))$, since $\prob{S\reachto u\mid\bddhi(\Psi)}=1$, we replace $\prob{S\reachto C_{\bddhi(\Psi)}(u)\mid \bddhi(\Psi)}$ in the above formula with $1$.
This replacement is reflected in case (i) of Definition~\ref{def:successive}.

What can be done when such a vertex $u$ does not exist?
In this case, we attempt to use the original TC.
We observe in Section~\ref{ssec:equiv} that if $u\in W_i$, we can recover $\Phi=\Phi_i^u(E^\prime)$ from $\Psi_i(E^\prime)$ with an appropriate transformation.
Thus, by fixing some $u\in C$ and letting $\Phi^u$ be a TC recovered from STC $\Psi$, we have the following case analysis:
If $e_i$ is present, $\prob{S\reachto C\mid\Psi}=\prob{S\reachto u\mid\Phi^u}$ equals $\prob{S\reachto u\mid\bddhi^u(\Phi^u)}$.
If $e_i$ is absent, $\prob{S\reachto C\mid\Psi}$ equals $\prob{S\reachto u\mid\bddlo^u(\Phi^u)}$.
This leads to the following decomposition.
\begin{lemma}
  \label{lem:rdp2}
  For $\Psi=\Psi_i(E^\prime)$ and SCC $C$ of $G_i(\Psi)$, suppose that successive SCCs do not exist.
  Let $u\in C$ and $\Phi^u$ be a TC recovered from STC $\Psi$. Then,
  \begin{align}
    \prob{S\reachto C\mid\Psi}=p_{e_i}\cdot\prob{S\reachto u\mid \bddhi^u(\Phi^u)}+(1-p_{e_i})\cdot\prob{S\reachto u\mid \bddlo^u(\Phi^u)}. \label{eq:rdp2}
  \end{align}
\end{lemma}
To compute the terms in the right-hand-side of (\ref{eq:rdp2}), we can use the previous algorithm in Section~\ref{sec:past}.
However, naively computing these terms for every $u\in V$ requires $\order{mn\omega_p^2\cdot 2^{\omega_p^2}}$ time, since this amounts to constructing a diagram of the previous algorithm for every $u$.

To overcome this, we use Lemma~\ref{lem:equiv1}.
Since we only need to compute $\prob{S\reachto v\mid\Phi^v}$ for the case $v\in A_i$, we can share identical TCs, as in Observation~\ref{obs:bottom}.
Accordingly, we do not need to construct different diagrams for every $v$ to compute $\prob{S\reachto v\mid \Phi^v}$, that is, if identical TCs $\Phi^v$ and $\Phi^w$ from different $v,w$ appear, we share them.
This enables us to compute all $\prob{S\reachto v\mid\Phi^v}$s in linear time.

\subsection{Procedure}
\label{ssec:procedures}
We formally describe the proposed method in Algorithm~\ref{alg:proposed}.
Lines 1--6 construct the diagram of STCs $\Psi$.
$\Psi.\bddf$ stores a pointer to $\bddf(\Psi)$.
We consider the successive SCCs in lines~7--11.
If a successive SCC of $C$ exists, $C.\bddf$ stores the pointer to it.
Otherwise, following the approach of Section~\ref{ssec:sketch}, we recover the original TC $\Phi^u$, where $u\in C$, and then carry out transition $\bddlo^u$ or $\bddhi^u$ on it.
$\trans^u(\Psi)$ recovers the original TC by removing every column $x$ satisfying $\Psi_{xu}=1$ from $\Psi$, that is, we remove columns whose indices can reach $u$.
If a pruning occurs by considering $\Phi=\bddf^u(\trans^u(\Psi))$, $C.\bddf$ stores $\bot$ or $\top$.
Otherwise, $C.\bddf$ stores a pointer to $\modifymap^{t}(\Phi)$, where $\modifymap^{t}(\cdot)$ changes the column index $u$ to $t$.
Here, $t$ is an imaginary vertex that is assumed to be contained within $A_i$.
Since Observation~\ref{obs:bottom} confirms that the transitions of identical TCs are identical, we standardize the column index $u$ by changing it to $t$.
We build a diagram of $\Phi_i^{t}(\cdot)$ in Lines 12--15.

Lines 17--19 compute $\Pmap[\Psi]=\prob{\Psi}$ for every $\Psi\in\mathcal{M}_i^\prime$ in the same manner as Algorithm~\ref{alg:previous}, and lines~20--22 compute $\Qmap[\Phi]=\prob{S\reachto t\mid\Phi}$ for every $\Phi\in\mathcal{N}_i^\prime$ in a bottom-up manner.
Lines~23--26 compute $\Rmap[\Psi,C]=\prob{S\reachto C\mid\Psi}$ for every $\Psi$ and every SCC $C$ of $G_i(\Psi)$.
Here, $\langle c \text{ ? } x \text{ : } y\rangle$ stands for a conditional operator that returns $x$ if $c$ is satisfied or $y$ otherwise.
When $C.\bddf$ is either an original TC $\Phi^t$ or $\bot,\top$, we use the $\Qmap$ value; otherwise, we use the $\Rmap$ value.
Finally, lines~27--31 compute every $\Resmap[v]=\prob{S\reachto v}$ using (\ref{eq:expansion4}).

\begin{algorithm}[!tb]
{\footnotesize
  $\mathcal{M}_1^\prime\leftarrow\{\Psi_1(\emptyset)\}$, $\mathcal{M}_i^\prime\leftarrow\{\}\ (i\geq 2)$, $\mathcal{N}_i^\prime\leftarrow\{\}\ (i\geq 1)$\;
  \For(\tcp*[f]{Diagram construction for $\Psi$}){$i \leftarrow 1\ \KwTo\ m$}{
    \ForEach{$\Psi\in\mathcal{M}_i^\prime$}{
      \ForEach{$\bddf\in\{\bddlo,\bddhi\}$}{
        \lIf{$\bddf(\Psi)=\bot$}{$\Psi.\bddf\leftarrow\bot$}
        \lElse{ $\mathcal{M}_{i+1}^\prime\leftarrow\mathcal{M}_{i+1}^\prime\cup\{\bddf(\Psi)\}$, $\Psi.\bddf\leftarrow\bddf(\Psi)\in\mathcal{M}_{i+1}^\prime$}
        \ForEach(\tcp*[f]{Compute successive SCCs}){$C$: SCC of $G_i(\Psi)$}{
          \lIf(\tcp*[f]{Including the cases $\bddf(C)=\bot,\top$}){$\bddf(C)\neq\emptyset$}{$C.\bddf\leftarrow\bddf(C)$}
          \Else(\tcp*[f]{Choose $u\in C$ arbitrarily}){
            \lIf{$\Phi\leftarrow\bddf^u(\trans^u(\Psi))=\bot$ or $\top$}{$C.\bddf\leftarrow \bot$ or $\top$}
            \lElse{ $\mathcal{N}_{i+1}^\prime\leftarrow\mathcal{N}_{i+1}^\prime\cup\{\modifymap^{t}(\Phi)\}$, $C.\bddf\leftarrow\modifymap^{t}(\Phi)\in\mathcal{N}_{i+1}^\prime$}
          }
        }
      }
    }
    \ForEach(\tcp*[f]{Diagram construction for $\Phi^t$}){$\Phi\in\mathcal{N}_i^\prime$}{
      \ForEach{$\bddf\in\{\bddlo,\bddhi\}$}{
        \lIf{$\bddf^{t}(\Phi)=\bot$ or $\top$}{$\Phi.\bddf\leftarrow\bot$ or $\top$}
        \lElse{ $\mathcal{N}_{i+1}^\prime\leftarrow\mathcal{N}_{i+1}^\prime\cup\{\bddf^{t}(\Phi)\}$, $\Phi.\bddf\leftarrow\bddf^{t}(\Phi)\in\mathcal{N}_{i+1}^\prime$}
      }
    }
  }
  $\Pmap[\Psi]\leftarrow1$ for $\Psi\in\mathcal{M}_1^\prime$, $\Pmap[\Psi]=0$ for $\Psi\in\mathcal{M}_i^\prime$ $(i\geq 2)$, $\Qmap[\bot]\leftarrow 0$, $\Qmap[\top]\leftarrow 1$\;
  \For(\tcp*[f]{Top-down DP}){$i \leftarrow 1\ \KwTo\ m$}{
    \ForEach(\tcp*[f]{$\Pmap[\Psi]$ stores $\prob{\Psi}$}){$\Psi\in\mathcal{M}_i^\prime$}{
      $\Pmap[\Psi.\bddlo]\pluseq (1-p_{e_i})\cdot\Pmap[\Psi]$, $\Pmap[\Psi.\bddhi]\pluseq p_{e_i}\cdot\Pmap[\Psi]$\;
    }
  }
  \For(\tcp*[f]{Bottom-up DP}){$i \leftarrow m\ \KwTo\ 1$}{
    \ForEach(\tcp*[f]{$\Qmap[\Phi]$ stores $\prob{S\reachto t\mid\Phi}$}){$\Phi\in\mathcal{N}_i^\prime$}{
      $\Qmap[\Phi]\leftarrow (1-p_{e_i})\cdot\Qmap[\Phi.\bddlo]+p_{e_i}\cdot\Qmap[\Phi.\bddhi]$\;
    }
    \ForEach{$\Psi\in\mathcal{M}_i^\prime$}{
      \ForEach(\tcp*[f]{$\Rmap[\Psi,C]$ stores $\prob{S\reachto C\mid\Psi}$}){$C$: SCC of $G_i(\Psi)$}{
        $\Rmap[\Psi,C]\leftarrow (1-p_{e_i})\cdot\langle\text{$\bddlo(C)\in\{\top,\bot,\emptyset\}$ ? }\Qmap[C.\bddlo]\text{ : }\Rmap[\Psi.\bddlo,C.\bddlo]\rangle$\\
        $\qquad\qquad\qquad+p_{e_i}\cdot\langle\text{$\bddhi(C)\in\{\top,\bot,\emptyset\}$ ? }\Qmap[C.\bddhi]\text{ : }\Rmap[\Psi.\bddhi,C.\bddhi]\rangle$\;
      }
    }
  }
  \ForEach(\tcp*[f]{Compute $\prob{S\reachto v}$ by using Eq. (\ref{eq:expansion4})}){$v\in V\setminus S$}{
    $\Resmap[v]\leftarrow 0$ \tcp*{$\Resmap[v]$ stores $\prob{S\reachto v}$}
    \ForEach(\tcp*[f]{Choose $i^\star$ arbitrarily s.t. $v\in W_{i^\star}$}){$\Psi\in\mathcal{M}_{i^\star}^\prime$}{
      $\Resmap[v]\pluseq\Pmap[\Psi]\cdot\langle\text{$v\in W_{i^\star}^{S\reachto}(\Psi)$ ? }1 \text{ : }\Rmap[\Psi,C_\Psi(v)]\rangle$
    }
  }
  \Return $\Resmap$
  \caption{Proposed algorithm for computing $\prob{S\reachto v}$ for every $v$}
  \label{alg:proposed}
}
\end{algorithm}

Here, we confirm the correctness of Algorithm~\ref{alg:proposed}.
\begin{lemma}
  \label{lem:correctness}
  After executing Algorithm~\ref{alg:proposed}, $\Resmap[v]$ equals $\prob{S\reachto v}$.
\end{lemma}
\begin{proof}
First, we confirm $\Pmap[\Psi]=\prob{\Psi}$.
Starting from $\Pmap[\Psi_1(\emptyset)]=\prob{\Psi_1(\emptyset)}=1$ for $\Psi_1(\emptyset)\in\mathcal{M}^\prime_1$, lines 17--19 compute every $\prob{\Psi}$ value for $\mathcal{M}^\prime_2,\ldots,\mathcal{M}^\prime_m$ according to (\ref{eq:pdp}).

Second, we confirm $\Qmap[\Phi]=\prob{S\reachto t\mid\Phi}$, but this is much simpler;
with the case analysis of whether $e_i$ is present or absent, we have
\begin{align}
  \prob{S\reachto t\mid \Phi}=p_i\cdot\prob{S\reachto t\mid\bddhi^t(\Phi)}+(1-p_i)\cdot\prob{S\reachto t\mid\bddlo^t(\Phi)}. \label{eq:qdp}
\end{align}
Starting from $\Qmap[\bot]=\prob{S\reachto t\mid\bot}=0$ and $\Qmap[\top]=\prob{S\reachto t\mid\top}=1$, lines 20--22 compute every $\prob{S\reachto t\mid\Phi}$ value for $\mathcal{N}^\prime_m,\mathcal{N}^\prime_{m-1},\ldots$ according to (\ref{eq:qdp}).

Third, we confirm $\Rmap[\Psi,C]=\prob{S\reachto C\mid\Psi}$.
If there exists $u\in C$ such that $u\in W_{i+1}$, i.e., a successive SCC exists, we can apply (\ref{eq:rdp1}).
Otherwise (i.e., a successive SCC does not exist), we can apply (\ref{eq:rdp2}).
Here, without considering the pruning, $C.\bddf$ stores a pointer to $\bddf(C)$ if a successive SCC exists or, otherwise, to $\modifymap^{t}(\bddlo^u(\Phi^u))$.
Consequently, lines 23--26 correctly compute $\Rmap[\Psi,C]$ according to (\ref{eq:rdp1}) if a successive SCC exists or to (\ref{eq:rdp2}) otherwise.
For the latter case, we use Lemma~\ref{lem:equiv1}; i.e., $\prob{S\reachto u\mid \bddlo^u(\Phi^u)}=\prob{S\reachto t\mid\modifymap^t(\bddlo^u(\Phi^u))}=\Qmap[\modifymap^t(\bddlo^u(\Phi^u))]$.

Finally, by choosing $i^\star$ such that $v\in W_{i^\star}$, $\Resmap[v]=\prob{S\reachto v}$ is correctly computed using (\ref{eq:expansion4}) through lines 27--30.
\end{proof}

\subparagraph*{Treatment of degree $1$ vertices.}
We finally describe how to deal with degree $1$ vertices.
By obtaining the $\prob{S\reachto v}$ values for all vertices with degree not less than $2$, we can easily obtain the $\prob{S\reachto v}$ values for degree $1$ vertices.
Let $v\in V\setminus S$ be a vertex whose degree is $1$.
Let $e_v$ be the unique edge incident to $v$.
If $e_v=(v,w)$ for some $w\in V$, i.e., $v$ is the tail of $e_v$, we have $\prob{S\reachto v}=0$ because $v$ has no incoming edges.
Otherwise (i.e., $v$ is the head of $e_v$), we have $\prob{S\reachto v}=p_{e_v}\cdot\prob{S\reachto w}$ because every path from some $s\in S$ to $v$ must go through $w$, thus $e_v$.
Since $G$ is connected and thus $w$ has degree at least $2$ except for the trivial case where $e_v$ is the only edge of $G$, we can obtain $\prob{S\reachto v}$ for $v$ whose degree is $1$ from the value of $\prob{S\reachto w}$.

\subsection{Complexity}
\label{ssec:complexity}
To analyze the complexity, we again use $\omega=\max_i|W_i|$.
First, we consider the diagram construction.
The key lemma is as follows.

\begin{restatable}{lemma}{LemTransition}
  \label{lem:transition}
  Given $\Phi=\Phi_i^v(E^\prime)$ for some $E^\prime\subseteq E_{<i}$, we can compute the transitions $\bddlo^v(\Phi)$ and $\bddhi^v(\Phi)$ in $\order{(|W_i|+|W_{i+1}|)^2}$ time.
  Similarly, given $\Psi=\Psi_i(E^\prime)$ for some $E^\prime\subseteq E_{<i}$, we can also compute $\bddlo(\Psi)$ and $\bddhi(\Psi)$ in $\order{(|W_i|+|W_{i+1}|)^2}$ time.
\end{restatable}

Before proceeding to the full proof, we describe the ideas for computing transitions.
For TC $\Phi$, we build a graph $G^\star=(W_i\cup\{S,v\},E^\star)$, where $S$ is regarded as an individual special vertex and $E^\star=\{(u,w)\mid \Phi_{uw}=1\}$.
By adding and removing edges and vertices according to the transition and computing the transitive closure, we can obtain the updated graph, whose adjacency matrix constitutes $\bddhi(\Phi)$ or $\bddlo(\Phi)$.
The transition of STC $\Psi$ can be computed in very similar way.

\begin{proof}[Proof of Lemma~\ref{lem:transition}]
  We show the procedure for computing the transition of $\Phi$ as follows.
  First, we build a graph $G^\star=(W_i\cup\{S,v\},E^\star)$, where $S$ is regarded as an individual special vertex and $E^\star=\{(u,w)\mid \Phi_{uw}=1\}$.
  From the definition of $\Phi$, $u\reachto_{(G,E^\prime)}w$ if and only if $u\reachto_{G^\star}w$ for any $u\in W_i\setminus W_i^{S\reachto}(E^\prime)\cup\{S\}$ and $w\in W_i\setminus W_i^{\reachto v}(E^\prime)\cup\{v\}$.
  Second, we add vertices in $W_{i+1}\setminus W_i$ to $G^\star$ as isolated vertices.
  If an added vertex $w$ is in $S$, we also add an edge $(S,w)$ to $G^\star$.\footnote{Although $w\reachto S$, we do not need to add $(w,S)$. This is because we are not concerned with the reachability to $S$. In addition, adding an edge $(w,S)$ wrongfully implies the reachability from $w$ to $u$, where $u$ is the vertex that can be reached from $S$.}
  Third, if we want to compute $\bddhi^v(\Phi)$, we add an edge $e_i=(u,w)$ to $G^\star$.
  At this time, we have $u\reachto_{(G,E^\prime)}w \iff u\reachto_{G^\star}w$ for any $u\in W_{i+1}\setminus W_{i+1}^{S\reachto}(E^\prime)\cup\{S\}$ and $w\in W_{i+1}\setminus W_{i+1}^{\reachto v}(E^\prime)\cup\{v\}$ when we compute $\bddlo^v(\Phi)$, or $u\reachto_{(G,E^\prime\cup\{e_i\})}w \iff u\reachto_{G^\star}w$ for any $u\in W_{i+1}\setminus W_{i+1}^{S\reachto}(E^\prime\cup\{e_i\})\cup\{S\}$ and $w\in W_{i+1}\setminus W_{i+1}^{\reachto v}(E^\prime\cup\{e_i\})\cup\{v\}$ when we compute $\bddhi^v(\Phi)$.
  Consequently, the desired TC can be computed from $G^\star$: We compute the transitive closure of $G^\star$ and then compute $W_{i+1}^{S\reachto}\coloneqq\{w\in W_{i+1}\mid S\reachto_{G^\star}w\}$ and $W_{i+1}^{\reachto v}\coloneqq\{w\in W_{i+1}\mid w\reachto_{G^\star}v\}$.
  The binary matrix $\Phi^\prime$ indexed by $(W_{i+1}\setminus W_{i+1}^{S\reachto}\cup\{S\})\times (W_{i+1}\setminus W_{i+1}^{\reachto v}\cup\{v\})$, the entries of which satisfy $\Phi^\prime_{uw}=1$ if and only if $u\reachto_{G^\star}w$, is exactly the desired TC $\bddlo^v(\Phi)$ or $\bddhi^v(\Phi)$.

  All procedures described above can be executed in quadratic time in the size of $G^\star$.
  Since the number of vertices in $G^\star$ is at most $\order{|W_i|+|W_{i+1}|}$, the overall cost can be bounded by $\order{(|W_i|+|W_{i+1}|)^2}$.

  For STC $\Psi$, the computation of transitions proceeds in almost the same manner as above, except that $G^\star$ is first prepared as $(W_i\cup\{S\},E^\star)$, where $E^\star=\{(u,w)\mid \Psi_{uw}=1\}$, and we do not need to compute $W_{i+1}^{\reachto v}$.
  The complexity remains the same $\order{(|W_i|+|W_{i+1}|)^2}$ time.

  Note that the SCCs and their correspondence to successive SCCs (Definition~\ref{def:successive}) can also be computed within this process.
  By computing the SCCs of $G_i(\Psi)$ and $G_{i+1}^\prime(\bddf(\Psi))$, the correspondence of Definition~\ref{def:successive} can be easily computed.
  Computing SCCs requires $\order{|W_i|^2+|W_{i+1}|^2}$ time, which is absorbed in the complexity.
\end{proof}

Similar to $\Phi_i^v$, we can compute the transition $\bddlo,\bddhi$ of $\Psi$ in $\order{\omega^2}$ by recomputing the transitive closure.
Note that we can also compute the SCCs and their successive SCCs within this process, as described in the proof of Lemma~\ref{lem:transition}.
Thus, the total time for diagram construction is bounded by $\order{\omega^2}\cdot\order{\sum_{i}(|\mathcal{M}_i|+|\mathcal{N}_i|)}$ time.
Here, we can bound the number of possible patterns on $\Phi_i^v$ and $\Psi_i$, i.e., $|\mathcal{N}_i|$ and $|\mathcal{M}_i|$, as follows.
\begin{restatable}{lemma}{LemPatternPhi}
  \label{lem:pattern_phi}
  $|\mathcal{M}_i|=|\{\Phi_i^v(E^\prime)\mid E^\prime\subseteq E_{<i}\}|$ is bounded by $\order{2^{|W_i|^2}}$.
\end{restatable}
\begin{proof}
  If we fix $W_i^{S\reachto}(E^\prime)$ and $W_i^{\reachto v}(E^\prime)$, some entries of $\Phi_i^v(E^\prime)$ are determined as follows: $\Phi_i^v(E^\prime)_{Sw}=1$ if $w\in W_i^{S\reachto}(E^\prime)$ or $0$ if $w\in W_i\setminus W_i^{S\reachto}(E^\prime)$, and $\Phi_i^v(E^\prime)_{uv}=1$ if $u\in W_i^{\reachto v}(E^\prime)$ or $0$ if $u\in W_i\setminus W_i^{\reachto v}(E^\prime)$.
  Therefore, the possible patterns on $\Phi_i^v(E^\prime)$ is bounded by $\order{2^{(|W_i|-|W_i^{S\reachto}(E^\prime)|)(|W_i|-|W_i^{\reachto v}(E^\prime)|)}}$ because $\Phi_i^v(E^\prime)$ is a binary matrix of size $(|W_i|-|W_i^{S\reachto}(E^\prime)|+1)\times (|W_i|-|W_i^{\reachto v}(E^\prime)|+1)$, and the entries of one row and one column are already determined.
  
  We conduct a case analysis of $k=|W_i^{S\reachto}(E^\prime)|+|W_i^{\reachto v}(E^\prime)|$ when fixing $W_i^{S\reachto}(E^\prime)$ and $W_i^{\reachto v}(E^\prime)$.
  When $k=0$, the bound is $\order{2^{|W_i|^2}}$.
  When $k=1$, the bound is $\order{2^{|W_i|^2-|W_i|}}$, and the number of patterns on $(W_i^{S\reachto}(E^\prime),W_i^{\reachto v}(E^\prime))$ is bounded by $\order{|W_i|}$.
  When $k\geq 2$, since $(|W_i|-|W_i^{S\reachto}(E^\prime)|)(|W_i|-|W_i^{\reachto v}(E^\prime)|)\leq (|W_i|-1)(|W_i|-1)=|W_i|^2-2|W_i|+1$ (supposing that $|W_i|\geq 3$), the number of possible patterns is bounded by $\order{2^{|W_i|^2-2|W_i|}}$.
  The number of possible patterns on $(W_i^{S\reachto}(E^\prime),W_i^{\reachto v}(E^\prime))$ is bounded by $2^{|W_i|}\cdot 2^{|W_i|}$ because $W_i^{S\reachto},W_i^{\reachto v}\subseteq W_i$.
  Thus, the total number of patterns on $\Phi_i^v(E^\prime)$ is bounded by $\order{2^{|W_i|^2}}+\order{|W_i|\cdot 2^{|W_i|^2-|W_i|}}+\order{2^{2|W_i|}}\cdot\order{2^{|W_i|^2-2|W_i|}}=\order{2^{|W_i|^2}}+\order{2^{|W_i|^2+\log|W_i|-|W_i|}}+\order{2^{|W_i|^2}}=\order{2^{|W_i|^2}}$.
\end{proof}

\begin{restatable}{lemma}{LemPatternPsi}
  \label{lem:pattern_psi}
  $|\mathcal{N}_i|=|\{\Psi_i(E^\prime)\mid E^\prime\subseteq E_{<i}\}|$ is bounded by $\order{2^{|W_i|^2}}$.
\end{restatable}
\begin{proof}
  We have a similar proof to that of Lemma~\ref{lem:pattern_phi}.
  When we fix $W_i^{S\reachto}(E^\prime)$, $\Psi_i(E^\prime)_{Sw}=1$ if $w\in W_i^{S\reachto}(E^\prime)$ or $0$ if $w\in W_i\setminus W_i^{S\reachto}(E^\prime)$.
  Therefore, the number of possible patterns on $\Psi_i(E^\prime)$ is bounded by $\order{2^{|W_i|(|W_i|-|W_i^{S\reachto}(E^\prime)|}}$.

  We again conduct a case analysis of $k=|W_i^{S\reachto}(E^\prime)|$.
  When $k=0$, the bound is $\order{2^{|W_i|^2}}$.
  When $k\geq 1$, the number of patterns on $\Psi_i(E^\prime)$ is bounded by $\order{2^{|W_i|^2-|W_i|}}$ and the number of possible patterns on $W_i^{S\reachto}(E^\prime)\subseteq W_i$ is $\order{2^{|W_i|}}$.
  Thus, the total number of patterns on $\Psi_i(E^\prime)$ is bounded by $\order{2^{|W_i|^2}}+\order{2^{|W_i|}}\cdot\order{2^{|W_i|^2-|W_i|}}=\order{2^{|W_i|^2}}$.
\end{proof}

Therefore, $|\mathcal{M}_i|=|\mathcal{N}_i|=\order{2^{\omega^2}}$, and the total time complexity for diagram construction can be bounded by $\order{m\omega^2\cdot 2^{\omega^2}}$.

Next, we consider the DP.
Each value can be computed in constant time.
The number of values to compute for $\Psi_i$ and $\Phi_i^v$ is bounded by $\order{|\mathcal{M}_i|+|W_i||\mathcal{M}_i|+|\mathcal{N}_i|}=\order{\omega\cdot 2^{\omega^2}+2^{\omega^2}}=\order{\omega\cdot 2^{\omega^2}}$.
Thus, the total complexity is $\order{m\omega\cdot 2^{\omega^2}}$.

Finally, computing each $\prob{S\reachto v}$ using (\ref{eq:expansion4}) requires $\order{|\mathcal{M}_i|}=\order{2^{\omega^2}}$ time.
Thus, it needs $\order{n\cdot 2^{\omega^2}}$ time in total for all $\prob{S\reachto v}$ values.
\begin{proposition}
  \label{prop:complexity1}
  Given directed graph $G$, seed set $S$, each edge's probability $p_e$ of presence, and the edge ordering $e_1,\ldots,e_m$, the probability $\prob{S\reachto v}$ for every vertex $v$ can be computed in $\order{(m+n)\cdot \omega^2\cdot 2^{\omega^2}}$ time in total.
\end{proposition}

Since we can generate an edge ordering with $\omega\leq \omega_p$ from a path decomposition of width $\omega_p$~\cite{inoue2016pathwidth} in linear time, we achieve Theorem~\ref{thm:main2}.
Moreover, with the linear-time path-decomposition algorithm for a bounded-pathwidth graph~\cite{bodlaender96tw}, we have the following.
\begin{theorem}
  \label{thm:main}
  Given directed graph $G$, seed set $S$, and each edge's probability $p_e$ of presence, suppose that $G$'s pathwidth is bounded by a constant.
  Then, $\prob{S\reachto v}$ for every vertex $v$ can be computed in $\order{m+n}$ time in total.
\end{theorem}

The previous algorithm~\cite{maehara17} requires $\order{m\omega^2\cdot 2^{\omega^2}}$ time for computing $\prob{S\reachto v}$, as discussed in Section~\ref{sec:past}.
Thus, it takes $\order{nm\omega^2\cdot 2^{\omega^2}}$ time for computing $\sigma(S)$.
In contrast, the proposed algorithm takes $\order{(n+m)\omega^2\cdot 2^{\omega^2}}$ time.
Therefore, in terms of time complexity, our algorithm reduces the dependence on the graph size from $\order{mn}$ to $\order{m+n}$ while the dependence on $\omega=\max_i|W_i|$ is kept the same.

\section{Discussion}
\label{sec:discussion}
Here, we point out the differences between the proposed algorithm and the existing algorithm that computes the probability of connection for undirected graphs~\cite{nakamura21}.
These algorithms basically share the same idea: In computing the probability of reachability (or connectivity in the undirected graphs) of every vertex, the previous top-down style algorithms, like that in Section~\ref{sec:past}, generate similar diagrams for different vertices $v$, so they attempt to share equivalent structures.
The algorithm for undirected graphs uses connected components (CCs) among frontier vertices to efficiently compute the conditional probabilities, similar to $\prob{S\reachto C\mid\Psi}$ in our algorithm.
The difference comes from the computation of conditional probability when a successive SCC or CC does not exist.
In undirected graphs, if a successive CC does not exist, the current CC is disconnected from any further vertex, so we immediately have the conditional probability $0$ (or $1$ in an exceptional case).
In contrast, in directed graphs, even if a successive SCC does not exist, the current SCC may still reach or be reached from further vertices because it can reach or be reached from the frontier vertices.
This makes the computation of the conditional probabilities difficult because we cannot immediately determine them when no successive SCC exists.
We resolve this issue by focusing on the equivalence of lower levels of diagrams (Lemma~\ref{lem:equiv1}).
This enables us to compute the conditional probabilities using a constant number of diagrams (STCs $\Psi$ and TCs $\Phi^t$ of imaginary vertex $t$) without breaking the linear complexity.

\section{Conclusion}
\label{sec:conclusion}
We proposed an algorithm for computing the influence spread under IC model exactly in $\order{(m+n)\omega_p^2\cdot 2^{\omega_p^2}}$ time.
For graphs with bounded pathwidth $\omega_p$, we improved the complexity from $\order{mn}$ to $\order{m+n}$, linear in the graph size, while the dependence on the pathwidth $\omega_p$ is kept the same.

Finally, we mention some future directions.
First, a more efficient exact evaluation of influence spread, suitable for the greedy algorithm~\cite{kempe03im}, should be investigated.
Chen et al.~\cite{chen09efficient} developed an algorithm that simultaneously approximates $\sigma(S\cup\{v\})-\sigma(S)$ for all $v\in V\setminus S$, which corresponds to one iteration of the greedy algorithm.
Nakamura et al.~\cite{nakamura23a} developed an algorithm for undirected graphs that exactly and simultaneously computes the expected number of connected vertices for all vertices.
This algorithm is promising for adoption with directed graphs as a way to simultaneously compute $\sigma(S\cup\{v\})$ for all $v\in V\setminus S$ exactly.
Second, we should investigate whether a linear time algorithm for computing $\sigma(S)$ is possible for graphs with bounded \emph{treewidth}.
A roadmap for achieving this may be to fit the diagram construction method based on tree decomposition~\cite{amarilli17circuit} to the computation of influence spread and then observe the equivalences.
However, such an improvement may not be straightforward due to the complicated nature of these algorithms.



\bibliography{mybib}

\end{document}